\pdfoutput=1
\documentclass[UKenglish]{llncs}
\usepackage[usenames, dvipsnames]{xcolor}
\usepackage[utf8]{inputenc}
\usepackage[english]{babel}
\usepackage{amsmath,amsfonts}
\usepackage[textwidth=1.1in]{todonotes}
\usepackage{cite}
\usepackage{mathtools}
\usepackage{tikz}
\usepackage{enumerate}
\usepackage[numbers,sectionbib]{natbib}
\usepackage[hyperindex,breaklinks]{hyperref}

\DeclarePairedDelimiter{\ceil}{\lceil}{\rceil}

\newtheorem{obs}[theorem]{Observation}

\title{Compact I/O-Efficient Representation of Separable Graphs and Optimal Tree Layouts\thanks{The work was supported by the Czech Science Foundation (GACR) project 17-10090Y "Network optimization"}}
\titlerunning{Compact I/O-Efficient Repres. of Sep. Graphs and Opt. Tree Layouts}

\author{Tomáš Gavenčiak, Jakub T\v{e}tek}
\institute{Department of Applied Mathematics, Charles University,\\ Malostransk\'{e} n\'{a}m. 25, CZ -- 11800 Prague, Czech Republic\\
	\email{\{gavento,jtetek\}@kam.mff.cuni.cz}
}

\begin{document}

\maketitle

\begin{abstract}
Compact and I/O-efficient data representations play an important role in efficient algorithm design, as memory bandwidth and latency can present a significant performance bottleneck, slowing the computation by orders of magnitude. While this problem is very well explored in e.g.~uniform numerical data processing, structural data applications (e.g.~on huge graphs) require different algorithm-dependent approaches. Separable graph classes (i.e.~graph classes with balanced separators of size $\mathcal{O}(n^c)$ with $c<1$) include planar graphs, bounded genus graphs, and minor-free graphs.

In this article we present two generalizations of the separator theorem, to partitions with small regions \emph{only on average} and to weighted graphs. Then we propose I/O-efficient succinct representation and memory layout for random walks in (weighted) separable graphs in the pointer machine model, including an efficient algorithm to compute them. Finally, we present a worst-case I/O-optimal tree layout algorithm for root-leaf path traversal, show an additive (+1)-approximation of optimal compact layout and contrast this with NP-completeness proof of finding an optimal compact layout.
\end{abstract}

\section{Introduction}

Modern computer memory consists of several memory layers that together constitute a memory hierarchy with every level further from the CPU being larger and slower~\cite{aggarwal1987hierarchical}, usually by more than an order of magnitude, e.g.~CPU registers, L1 -- L3 caches, main memory, disk drives etc. In order to simplify the model, commonly only two levels are considered at once, called \emph{main memory} and \emph{cache} of size $M$. There, the main memory access is block-oriented, assuming unit time for reading and writing of a block of size $B$, making random byte access very inefficient. While some I/O-efficient algorithms need to know the values of $B$ and $M$ (generally called \emph{cache-aware})\cite{bib:cache_aware}, \emph{cache-oblivious algorithms}\cite{bib:cache_oblivious} operate efficiently without this knowledge.

Computations that process medium to large volumes of data therefore call for space-efficient data representations (to utilize the memory capacity and bandwidth) and strongly benefit from optimized memory access patterns and layouts (to utilize the data in fast caches and read-ahead mechanisms).
While this area is very well explored in e.g.~numerical data processing and analysis (e.g.~\cite{kowarschik2003overview}), structural data applications (e.g.~huge graphs) require different and application-dependent approaches. We describe a representations to address these issues in separable graphs and trees.

\emph{Separable graphs} satisfy the \emph{$n^c$-separator theorem} for some $c<1$, shown for planar graphs in 1979 by \citet{bib:plan_sep} (with $c=1/2$), where every such graph on $n$ vertices has a vertex subset of size $\mathcal{O}(n^c)$ that is a $2/3$-balanced separator (i.e. it separates the graph into two subgraphs each having at most $2/3$-fraction of vertices). These graphs not only include planar graphs~\cite{bib:plan_sep} but also bounded genus graphs~\cite{GILBERT1984391} and minor-free graph classes in general~\cite{bib:minor_sep}. Small separators are also found in random graph models of small-world networks (e.g.\ geometric inhomogeneous random graphs by \citet{DBLP:conf/esa/BringmannKL17} have sublinear separators w.h.p.\ for all subgraphs of size $\Omega(\sqrt{\log n})$). Some graphs which come from real-world applications are also separable, such as the road network graphs~\cite{wang2014costefficient,bib:roads}. Separable graph classes have linear information entropy (i.e.\ a separable class can contain only $2^{\mathcal{O}(n)}$ graphs of size $n$) and have efficient representations using only $\mathcal{O}(1)$ bits per vertex on average\cite{blandford2003compact} and therefore utilize the memory capacity and bandwidth very efficiently.

\medskip
This paper is organized as follows: Sections~\ref{sec:related} and~\ref{sec:contribution} give an overview of the prior work and our contribution. Section~\ref{sec:preliminaries} recalls used concepts and notation. Section~\ref{sec:random-walks} contains our results on random walks in separable graphs. Section~\ref{sec:sep_hierarchy} generalizes the separator theorem. Section~\ref{sec:trees} discusses the layout of trees.

\subsection{Related work}\label{sec:related}

\citet{turan1984succinct} introduced a succinct representation\footnote{A succinct (resp. compact) data representation uses $H+o(H)$ (resp. $\mathcal{O}(H)$) bits where $H$ is the class information entropy.} of planar graphs, \citet{blandford2003compact} introduced compact representations for separable graphs and \citet{blelloch2010succinct} presented a succinct representation of separable graphs. However, none of those representations is cache-efficient (or can be easily made so). Analogous representations for general graphs suffer similar drawbacks~\cite{naor1990succinct, farzan2013succinct}.

\citet{agarwal1998efficient} developed a representation of planar graphs allowing I/O-efficient path traversal, requiring $\mathcal{O}(K/\log B)$ block accesses\footnote{Note that $\Omega(K/\log B)$ blocks may be required even for trees. Standard graph representation would access $\mathcal{O}(K)$ blocks.} for arbitrary path of length $K$. This has been extended to a succinct planar graph representation by \citet{dillabaugh2017efficient} with the same result for arbitrary path traversal. It appears unlikely that the representation of \cite{dillabaugh2017efficient} could be easily modified to match the I/O complexity $\mathcal{O}(K/B)$ of our random-walk algorithm due to their use of a global indexing structure.

\smallskip
\citet{dillabaugh2012succinct} describes a succinct data structure for trees that uses $\mathcal{O}(K/B)$ I/O operations for leaf-to-root path traversal. For root-to-leaf traversal, they offer a similar but only compact structure.

\smallskip
Among other notable I/O-efficient algorithms, \citet{maheshwari2002optimal} develop I/O-efficient algorithms for computing vertex separators, shortest paths and several other problems in planar and separable graphs. \citet{bib:cache_shortest_path} extends this to a cache-oblivious algorithm for planar shortest paths. While there are representations even more efficient than succinct (e.g.\ implicit representations, which use only $\mathcal{O}(1)$ bits more than the class information entropy, see \citet{kannan1992implicit} for an implicit graph representation), these do not seem to admit I/O-efficient access.

\smallskip
Random walks on graphs are commonly used in Monte Carlo sampling methods, among others in Markov Chain Monte Carlo methods for inference on graphical models~\cite{gamerman2006markov}, Markov decision process (MDP) inference and even in partial-information game theory algorithms~\cite{lanctot2009monte}.

\subsection{Our contribution}\label{sec:contribution}

\textbf{Random walks on separable graphs.} We present a compact cache-oblivious representation of graphs satisfying the $n^c$ edge separator theorem. We also present a cache-oblivious representation of weighted graphs satisfying weighted $n^c$ edge separator theorem, where the transition probabilities depend on the weights. The representations are I/O-efficient when performing random walks of any length on the graph, starting from a vertex selected according to the stationary distribution and with transition probabilities at each step proportional to the weights on the incident edges, respectively choosing a neighbor uniformly at random for the unweighted compact representation.

Namely, if every vertex contains $q$ bits of extra (user) information, the representation uses $\mathcal{O}(n\log(q+2))+qn$ bits and a random path of length $K$ (sampled w.r.t.\ edge weights) uses $\mathcal{O}(K/(\frac{Bw}{(1+q)})^{1-c})$ I/O operations with high probability.

The graph representation is compact (as the structure entropy including the extra bits is $\Theta((q+1)n)$. The amount of memory used for the representation of the graph is asymptotically strictly smaller than the memory used by the user data already for the common case of $q=\Theta(w)$, in which case only $\mathcal{O}(K/B^{1-c})$ I/O operations are used. For $q=\mathcal{O}(1)$, the representation uses $\mathcal{O}(n)$ bits.

In contrast with previous I/O-efficient results for planar graphs, our representation is only compact (and not succinct) but works for all separable graph classes, is cache-oblivious (in contrast to only cache-aware in prior work), and, most importantly, comes with a much better bound on the number of I/O operations for randomly sampled paths (order of $\mathcal{O}(K/B^{1-c})$ rather than $\mathcal{O}(K/\log B)$). 

\medskip \noindent
\textbf{Fast tree path traversal} is a ubiquitous requirement for tree-based structures used in external storage systems, database indexes and many other applications. With Theorem~\ref{thm:tree_layout}, we present a linear time algorithm to compute a layout of the vertices in memory minimizing the worst-case number of I/O operations for leaf-to-root paths in general trees and root-to-leaf paths in trees with unit vertex size. We show an additive (+1)-approximation of an optimal \emph{compact} layout (i.e.\ one that fully uses a consecutive block of memory) and show that finding an optimal compact layout is $NP$-hard.

The above layout optimality is well defined assuming unit vertex size, an assumption often assumed and satisfied in practice. Using techniques from Section~\ref{sec:random-walks} we can turn the layout into a compact representation using $\mathcal{O}(n)$ bits of memory, requiring at most $OPT_L$ I/O operations for leaf-to-root paths in general trees and root-to-leaf paths in trees of fixed degree where $OPT_L$ is the I/O complexity of the optimal layout, i.e.\ I/O-optimal layout with the vertices using any conventional vertex representation with $\Theta(w)$ bits for inter-vertex pointers. See Theorem~\ref{thm:tree-compact}.

Compared to previous results \cite{dillabaugh2012succinct}, our representation is compact and we present the exact optimum over all layouts while they provide the asymptotic optimum $\mathcal{O}(K/B)$. However, this does not guarantee that our representation has lower I/O complexity, since our notion of optimality only considers different layouts with each vertex stored by a structure of unit size.

\medskip \noindent
\textbf{Separable graph theorems.} We prove two natural generalizations of the separator theorem (Theorem \ref{thm:sep_hierarchy}) and show that their natural joint generalization does not hold by providing a counterexample (Theorem \ref{thm:false}). The Recursive Separator Theorem involves graph partitions coming from recursive applications of the Separator Theorem. Let $r$ and $\bar r$ denote the maximum and average size of a region in the partition, respectively. We prove stronger bound on number of edges going between regions -- $\mathcal{O}(\frac{n}{\bar r^{1-c}})$ instead of $\mathcal{O}(\frac{n}{r^{1-c}})$. The second generalization is for weighted graphs, showing that $n$ in the bound $\mathcal{O}(\frac{n}{r^{1-c}})$ can be replaced by the total weight $W$ to get $\mathcal{O}(\frac{W}{r^{1-c}})$. We show that the bound $\mathcal{O}(\frac{W}{\bar r^{1-c}})$ does not hold in general by providing a counterexample.

\section{Preliminaries}\label{sec:preliminaries}

Throughout this paper, we use standard graph theory notation and terminology as in \citet{bib:ballobas}. We denote the subtree of $T$ rooted in vertex $v$ by $T_v$, the root of tree $T$ by $r_T$ and the set of children of a vertex $v$ as $\delta(v)$. All the logarithms are binary unless noted otherwise.

We use standard notation and results for Markov chains as introduced in the book by \citet{bib:prob_book} (chapter 11) and mixing in Markov chains, as introduced in the chapter on mixing times in a book by \citet{bib:mixing_chapter}.

\subsection{Separators}\label{separators}
Let $S$ be a class of graphs closed under the subgraph relation. We say that $S$ satisfies the vertex (edge) $f(n)$-separator theorem iff there exist constants $\alpha < 1$ and $\beta > 0$ such that any graph in $S$ has a vertex (edge) cut of size at most $\beta f(n)$ that separates the graph into components of size at most $\alpha n$.
We define a weighted version of vertex (edge) separator theorem, which requires that there is a balanced vertex (edge) separator of total weight at most $\beta \frac{f(n)}{n}W$, where $W$ is the sum of weights of all the edges. Note that these definitions make sense even for directed graphs. $f(n)$-separator theorem without explicit statement whether it is edge or vertex separator, means $f(n)$ vertex separator theorem.

Many graphs that arise in real-world applications satisfy $n^c$ vertex or edge separator theorem.

It has been extensively studied how to find balanced separators in graphs. In planar graphs, a separator of size $\sqrt{n}$ can be found in linear time \cite{bib:plan_sep}. Separators of the same size can be found in minor-closed families in time $\mathcal{O}(n^{1+\epsilon})$ for any $\epsilon > 0$ \cite{bib:minor_sep}. A balanced separator of size $n^{1-1/d}$ can be found in finite-element mesh in expected linear time \cite{bib:mesh_sep}. Good heuristics are known for some graphs which arise in real-world applications, such as the road network \cite{bib:roads}. A poly-logarithmic approximation which works on any graph class is known \cite{bib:sep_approx}. A poly-logarithmic approximation of the separators will be sufficient to achieve almost the same bounds in our representation (differing by a factor at most poly-logarithmic in $B$). 

We define a \textit{recursive separator partition} to be a partition of vertex set of a graph, obtained by the following recursive process. Given a graph $G$, we either set the whole $V(G)$ to be one set of the partition or do the following: 
\begin{enumerate}
\item Apply separator theorem. This gives us partition of $V(G)$ into two sets $A,B$ from the separator theorem.
\item Recursively obtain recursive separator partitions of $A$ and $B$.
\item Return the union of the partitions of $A$ and $B$ as the partition of $V(G)$.
\end{enumerate}
We call the sets in a recursive separator partition regions.\\
If there is an algorithm that computes balanced separator in time $\mathcal{O}(f(n))$, there is an algorithm that computes recursive separator partition with region size $\Theta(r)$ in time $\mathcal{O}(f(n) \log n)$ for any $r$. A stronger version called $r$-division can be computed in linear time on planar graphs \cite{bib:linear_r-division}.

\subsection{I/O complexity}
For definitions related to I/O complexity, refer to \citet{bib:demaine_cache_intro}. We use the standard notation with $B$ being the block size and $M$ the cache size. Both $B$ and $M$ is counted in words. Each word has $w$ bits and it is assumed that $w \in \Omega(\log n)$.

\section{Representation for Random Walks}\label{sec:random-walks}

In this section, we present our cache-oblivious representation of separable graphs optimized for random walks and related results.

\begin{theorem}\label{thm:separable}
Let $G$ be a graph from a graph class satisfying the $n^c$ edge separator theorem where every vertex contains $q$ extra bits of information. Then there is a cache-oblivious representation of $G$ using $\mathcal{O}\left(n\log (q+2)\right)+qn$ bits in which a random walk of length $k$ starting in a vertex sampled from the stationary distribution uses in expectation $\mathcal{O}\big(k/\big(\frac{Bw}{(1+q)}\big)^{1-c}\big)$ I/O operations. Moreover, such representation can be computed in time $\mathcal{O}(n^{1+\epsilon})$ for any $\epsilon > 0$.
\end{theorem}

For other random walks and weighted graphs where the transition probabilities are proportional to the random walk stationary distribution, we can show a weaker result. Namely, we can no longer guarantee a compact representation.

\begin{theorem} \label{thm:cache_random}
Let $M$ be any Markov chain of random walks on a graph $G$ and assume $M$ has a unique stationary distribution $\pi$. Assume $G$ satisfies the $n^c$ edge separator theorem with respect to the edges-traversal probabilities in $\pi$. Let $M'$ be a Markov chain of random walks on $G$ with transition probabilities proportional to $M$, e.g. $\pi'(e)=\Theta(\pi(e))$. Then there is a layout of vertices of $G$ into blocks with $\Theta(B)$ vertices each such that a random walk in $M'$ of length $k$ crosses memory block boundary in expectation $\mathcal{O}(k/B^{1-c})$ times.
\end{theorem}

Note that this gives an efficient memory representation when $N_G(v)$ and the probabilities on incident edges can be represented by (or computed from) $\mathcal{O}(1)$ \emph{words}, which is the case for bounded degree graphs with some chains $M'$. We also note that such partially-implicit graph representations are present in the state graphs of some MCMC probabilistic graphical model inference algorithms.

\medskip
Additionally, we present a result on the concentration of the number of I/O operations which applies to both Theorems~\ref{thm:separable} and \ref{thm:cache_random}.

\begin{theorem} \label{thm:concentration}
Let $G$ be a fixed graph, $t_{mix}$ the mixing time of $G$ and $X$ the number of edges going between blocks crossed during the random walk. Then the probability that $(1-\delta)E(X) \leq X \leq (1+\delta)E(x)$ does not hold is $\mathcal{O}\big(m e^{-c'\frac{\delta^2nB^{c-1}}{m}}\big)$ for some value $c'$ and $m=t_{mix} \log(n^2/E(X_1))$, where the variable $X_i$ indicates if the walk crossed an edge between two different blocks in step $i$.
\end{theorem}


\medskip \noindent
The following lemma is implicit in \cite{blandford2003compact}, as the authors use the same layout to get compact representation of separable graphs and they use the following property.

\begin{lemma}[\citet{blandford2003compact}] \label{lem:small_differences}
If $\pi$ in Theorem \ref{thm:cache_random} gives the same traversal probability to all edges, the representation induces a vertex order $l:V\to 1\dots n$ such that $\sum_{e=uv\in E}{\log |l(u)-l(v)|} = \mathcal{O}(n)$.
\end{lemma}

\subsection{Proofs of Theorems~\ref{thm:separable}~--~\ref{thm:concentration}}



\begin{proof}[Proof of Theorem~\ref{thm:separable}]
%

Since the stationary distribution on an undirected graph assigns equal probability to every edge, we can apply Lemma~\ref{lem:small_differences} on $G$ to obtain vertex ordering $r:V\to 1\dots n$ such that $\sum_{e=uv\in E_G}\log |r(u)-r(v)|=\mathcal{O}(n)$. We could therefore compactly store the edges as variable-width vertex order differences (offsets). However, it is not straightforward to find the memory location of a given vertex when a variable-width encoding is used. To avoid an external (and I/O inefficient) index used in some other approaches, we replace the edge offset information with relative bit-offsets, directly pointing to the start of the target vertex, using Theorem~\ref{thm:expand} on the edge offsets. We expand the representation by inserting the $q$ bits of extra information to every vertex, adjusting the pointers and thus widening each by $\mathcal{O}(\log q)$ bits.

To prove the bound on I/O complexity, we use the same argument as in the proof of Theorem~\ref{thm:cache_random}. Average of $\mathcal{O}(1 + q)$ bits is used for representation of single vertex and, therefore, average of $\Theta(\frac{Bw}{q+1})$ vertices fit into one cache line. By Theorem~\ref{thm:sep_hierarchy}, part \ref{itm:first}, the total probability on edges going between memory blocks is $\mathcal{O}(1/\frac{Bw}{q+1})$. Again, by linearity of expected value, this proves the claimed I/O complexity.

Compact representation as in Theorem~\ref{thm:expand} can be computed in the claimed bound, as is shown in Theorem \ref{thm:rep_construction}.
\qed 
\end{proof}




\begin{proof}[Proof of Theorem~\ref{thm:cache_random}]
We use the following recursive layout. Let $S$ be an edge separator with respect to edge-traversal probabilities in $\pi$. Then $S$ partitions $G$ into two subgraphs $X$ and $Y$. We recursively lay out $X$ and $Y$ and concatenate the layouts. Note that $X$ and $Y$ are stored in memory contiguously. At some level of recursion, we get partition into subgraphs represented by between $\epsilon B$ and $B$ words for $\epsilon > 0$ constant. We call these subgraphs block regions. Since the average degree in graphs satisfying $n^c$ edge separator theorem is $\mathcal{O}(1)$ \cite{bib:separable_degrees}, the average vertex representation size is also $\mathcal{O}(1)$ and the average number of vertices in a block region is, therefore, $\Theta(B)$. It follows from Theorem~\ref{thm:sep_hierarchy}, part \ref{itm:second}, that the total probability on edges going between block regions is $\mathcal{O}(1/B^{1-c})$. From linearity of expectation, $\mathcal{O}(1/B^{c-1})$-fraction of steps in the random walk cross between block regions in expectation. Moreover, each of the block regions in the partition is stored in $\mathcal{O}(1)$ memory blocks, which proves the claimed bound on I/O complexity.
\qed 
\end{proof} 

\begin{proof}[Proof of Theorem~\ref{thm:concentration}]
Let $X$ be the number of edges crossed during the random walk that go between blocks. We are assuming that there is at least one edge going between two blocks in the graph.

We choose $\delta' = \sqrt{\frac{3}{4}} \delta$ (arbitrary constant $c'' < 1$ would work). Note that $m$ is a number of steps, after which the probabilities on edges differ from those in stationary distribution by at most $E(X_1)/n^2$, regardless from what distribution we started the random walk since $t_{mix}(\epsilon) \leq \ceil{\log \epsilon^{-1}}t_{mix}$ \cite{bib:mixing_chapter}. This means that the probability that an edge going between two blocks is crossed after $m$ steps differs by at most $\frac{1}{n}$-fraction from the probability in stationary distribution.

Let $X_i$ be indicator random variable that is $1$ iff the random walk crosses edge going between blocks in step $i$. We consider the following sets of random variables $S_i = \{X_j|X_{j-m}: j \mod m\} = i\}$ for $1 \leq i \leq m$ (not conditioning on variables with nonpositive indices). Note that the random variables in each of sets $S_i$ are independent and $(1-\frac{1}{n})E(X_j) \leq E(X_j | X_{j-m}) \leq (1+\frac{1}{n})E(X_j)$, as mentioned above. Let $\mu_i$ be $E(\sum_{X \in S_i}X)$ and $\mu = E(\sum_{i}{\sum_{X \in S_i}{X}})$. Note that $\mu_i \in \Theta(nB^{c-1}/m)$ for each $i$. By applying the Chernoff inequality, we get that the following bounds hold for all $n \geq n_0$ for some $n_0$ for each $i$:
$$P\Big(\sum_{X \in S_i}{X} \geq (1+\delta')\mu_i\Big) \leq e^{-\frac{\delta'^2\mu_i}{3}} = e^{-\frac{\delta^2\mu_i}{4}}$$
$$P\Big(\sum_{X \in S_i}{X} \leq (1-\delta')\mu_i\Big) \leq e^{-\frac{\delta'^2\mu_i}{2}} \leq e^{-\frac{\delta^2\mu_i}{4}}$$
The probability that there exists $i$ such that either $\sum_{X \in S_i}{X} \geq (1+\delta')\mu_i$ or $\sum_{X \in S_i}{X} \leq (1-\delta)\mu_i$ is by the union bound for some value of $c'$ at most the following:
$$2\ceil{\log(n/E(X_1))}t_{mix}e^{-\frac{\delta^2\mu}{4m}} \in \mathcal{O}(m e^{-c'\frac{\delta^2nB^{c-1}}{m}})$$

Note that $\mu_i$ converges to $|S_i|E(X_1)$, which is the value that we are showing concentration of $\sum_{X \in S_i}X$ around. The asymptotic bound on the probability follows.
\qed 
\end{proof}

\subsection{Expanding relative offsets to relative bit-offsets}\label{sec:bit-offsets}

Having the edges of a graph encoded as relative offsets to the target vertex and having these numbers encoded by a variable-length encoding, we need a way to find the exact location of the encoded vertex. Others have used a global index for this purpose but this is generally not I/O-efficient.

Our approach encodes the relative offsets as slightly wider numbers that directly give the relative bit-index of the target. However, this is not straightforward as expanding just one relative offset to a relative bit-offset can make other bit-offsets (spanning over this value) larger and even requiring more space, potentially cascading the effect.

Note that one simple solution would be to widen every offset representation by $\Theta(\log\log N)$ bits where $N$ is the total number of bits required to encode all the $n$ offsets, yielding $N+n*\mathcal{O}(\log\log N)$ encoding. %
$\log n$ bits are sufficient to store each offset. Therefore, by expanding the offsets, they increase at most $\log n$ times. By adding $\log( 2 \log n)$ bits, we can encode increase of offsets by factor of up to $2 \log n \geq \log n + \log( 2 \log n)$.

However, we propose more efficient encoding with the following theorem. We interpret the numbers $a_i$ as relative pointers, $i$-th number pointing to the location of the $(i+a_i)$-th value. In the proof, we use a dynamic width \emph{gamma number encoding} in the form $[(\mathrm{sign})B_00B_10B_20\dots B_i1]$, where $2i+1$-th bit encodes whether $B_i$ is the last bit encoded.

\begin{theorem}\label{thm:expand}
    Let $a_1\dots a_n$ be a sequence of numbers such that $-i\leq a_i\leq n-i$ and $\sum_{i=0}^{n}\log |a_n|=m$. Then there are $n$-element sequences $\{w_i\}$ (the encoded bit-widths) and $\{b_i\}$ (the bit-offsets) of numbers such that for all $1\leq i\leq n$, $w_i\geq 2\log|b_i|+1$ (i.e.\ $b_i$ can be gamma-encoded in $w_i$ bits), $P(i)+w_i=P(i+a_i)$ where $P(j):=\sum_{i=1}^{j-1} w_i$ (so $w_i$ is a relative bit-offset of encoded position $i+a_i$) and $\sum_{i=1}^{n} w_i = \mathcal{O}(m+n)$.
\end{theorem}
\begin{proof}
There are certainly \emph{some} non-optimal valid choices for $w_i$'s and $b_i$'s, and we can improve upon them iteratively by shrinking $w_i$'s to fit gamma-encoded $b_i$ with sign (i.e.\ $w_i=1+2\log|b_i|$), which may, in turn, decrease some $b_i$'s. Being monotonic, this process certainly has a fixpoint $\{b_i\}_i$ and $\{w_i\}_i$ and we assume arbitrary such fixpoint.

Let $C<1$ and $D>1$ be constants to be fixed below. Denote $v_i=\log|a_i|$ and $R_i=\{i\dots i+a_i-1\}$ (resp. $\{i+a_i\dots i-1\}$ when $a_i<0$). Intuitively, when expanding offsets $a_x$ to bit offsets $b_x$, it may happen that $R_x$ contains $y$ with $w_y\gg a_x$, forcing $w_x\gg v_x$. We amortize such cases by distributing "extra bits" to such "smaller" offsets.

Let $x\prec y\iff y\in R_x \wedge v_x\leq C\log w_y \wedge v_x>D$ and let $x^\uparrow=\mathrm{arg\ max}_{y\succ x} w_y$ (or undefined if there is no such $y$) and let $y^\downarrow=\{x|y\in x^\uparrow\}$. Observe that $|y^\downarrow|\leq 2\cdot 2^{C\log w_y}=2w_y^C$ since all $x\in y^\downarrow$ have $|a_x|\leq 2^{v_x}\leq w_y^C$. We also note that $y=x^\uparrow$ implies $w_x<w_y$ since $w_y\leq w_x$ would imply $b_x\leq |a_x|w_x$ and $w_x>2^{v_x/C}$ leading to $w_x\leq v_x + \log w_x$ and $2^{v_x/C}<w_x\leq 2v_x$, which gives the desired contradiction with $D$ large enough (depending only on $C$).

We will distribute the extra bits starting from the largest $w_i$'s. Every $y$ uses $w_y$ bits for its encoding and distributes another $w_y$ bits to $y^\downarrow$. Let $r_x=w_{x^\uparrow}/|(x^\uparrow)^\downarrow|\geq\frac{1}{2}w_{x^\uparrow}^{1-C}$ be the number of extra bits received from $x^\uparrow$ in this way.

For every offset $x$ we use $10v_x+2D$ bits and the received bits $r_x$. Since the received bits are accounted for in other offsets, this uses $\sum_{i=1}^n 10v_x+D=10m+\mathcal{O}(n)$ bits in total. Therefore we only need to show that the number of bits thus available at $x$ is sufficient, i.e.\ that $2w_x \leq r_x+10v_x+2D$ (one $w_x$ to represent $b_x$, one to distribute to $x^\downarrow$).

Now either there is $y=x^\uparrow$ and we have $b_x\leq |a_x| w_y$ so $w_x\leq 1 + 2 v_x + 2 \log w_y$ and noting that for large enough $D$ only depending on $C$: $2\log w_y\leq \frac{1}{4} w_y^{1-C}+D\leq \frac{1}{2}r_x+D$, so we obtain $w_x \leq \frac{1}{2}r_x+5v_x+2D$ as desired.

On the other hand, undefined $x^\uparrow$ implies that $\forall y\in R_x: w_y \leq 2^{v_x/C}$. Therefore $b_x\leq |a_x| 2^{v_x/C}$ and $w_x\leq 1 + 2 v_x + 2 v_x / C = 1 + (2 + 2/c) v_x$. Now we may fix $C=2/3$, obtaining $w_x \leq 5 v_x+D$ as required for $D\geq 1$. This finishes the proof for any fixpoint $\{b_i\}_i$ and $\{w_i\}_i$.
\qed 
\end{proof}

\smallskip
The algorithm from the beginning of the proof can be shown to run in polynomial time. We start with e.g.\ $w_i=w_0=1+4\log n$ and $b_i=\operatorname{sign}(a_i)\sum_{j\in R_i}w_j$. Then we iteratively update $w_i:=1+2\lceil\log b_i\rceil$ and recompute $b_i$ as above. Since every iteration takes $\mathcal{O}(n^2)$ time and in every iteration at least one $w_i$ decreases, the total time is at most $\mathcal{O}(n^3\log n)$. In the following section, we show an algorithm that computes a representation with the same asymptotic bounds, running in time $\mathcal{O}(n^{1+\epsilon})$ for any $\epsilon > 0$.

\subsubsection{Constructing the compact representation}

In this section, we use notation defined in section \ref{sec:bit-offsets}, specifically $R_e$ and $b_e$. Recall that $R_e$ is the set of edges of $G$ spanned by the edge $e$ in the representation and $b_e$ is the relative offset of edge $e$ in the (expanded) representation).
Let $G$ be the graph we want to represent. We assume that $G$ satisfies the $n^c$ edge separator theorem. 

We find a representation using $\mathcal{O}(n \log \log n)$ bits, as mentioned above by expanding all pointers and then modify it to make it compact.

We define a directed graph $H$ on the set $E(G)$ with arc going from $v$ to $u$ iff $v \in R_u$. 
Let us fix a recursive separator hierarchy of $G$. 
We call $l(e)$ the level of recursion on which the edge $e$ is part of the separator.
We define a graph $H_{\leq k}$ to be the subgraph of $H$ induced by vertices corresponding to edges of $G$ which appear in the recursive separator hierarchy in a separator of subgraph of size at most $k$.

The following lemma will be used to bound the running time of the algorithm:
\begin{lemma}
The maximum out-degree of $H_{\leq n^{c'}}$ is $n^{c*c'}$. For any fixed $c' > 0$, $|H \setminus H_{\leq n^{c'}}| \in n^{1-\epsilon'}$ where $\epsilon' > 0$ is some constant depending only on $c$ and $c'$.
\end{lemma}
\begin{proof}
We first prove that maximum out-degree of $H$ is $\mathcal{O}(n^c)$. 

There are $\mathcal{O}(n^c)$ edges $e \in G$ with $l(e) = 1$ spanning any single vertex. The number of edges $e$ spanning some vertex with $l(e) = k$ decreases exponentially with $k$, resulting in a geometric sequence summing to $\mathcal{O}(n^c)$.

The maximum out-degree of $H_{\leq n^{c'}}$ is the same as that of graph $H'$ corresponding to a subgraph of $G$ of size at most $n^{c'}$. Maximum out-degree of $H_{\leq n^{c'}}$ is, therefore, $\mathcal{O}(n^{c*c'})$.

The number of vertices in $H \setminus H_{\leq n^{c'}}$ is equal to the number of edges in $G$ going between blocks of size $\Theta(n^{c'})$. This number is, by Theorem \ref{thm:sep_hierarchy}, equal to $n/n^{c'(1-c)}$, which is $\mathcal{O}(n^{1-\epsilon})$ for some $\epsilon' > 0$.
\qed 
\end{proof}

\begin{theorem}\label{thm:rep_construction}
Given a separator hierarchy, the representation from Theorem \ref{thm:separable} can be computed in time $\mathcal{O}(n^{1+\epsilon})$ for any $\epsilon > 0$.
\end{theorem}
\begin{proof}
We first describe an algorithm running in time $\mathcal{O}(n^{1+c} \log \log n)$, where $c$ is the constant from the separator theorem, and then improve it.

Just as in the proof of Theorem \ref{thm:expand}, $b_v$ denotes the relative offset of edge $v$ in the representation. We store a counter $c_v$ for each vertex $v \in H$ equal to the decrease of $b_v$ required to shrink its representation by at least one bit. That is, $c_v = b_v - \lfloor b_v \rfloor_{2^k} + 1$, where $\lfloor i \rfloor_{2^k}$ is $i$ rounded down to closest power of two. When we shrink the representation of edge corresponding to vertex $v \in H$, we have to update counters $c_u$ for all $u$, such that $vu \in E(H)$. Since the out-degree of $H$ is $\mathcal{O}(n^c)$, the updates take $\mathcal{O}(n^c)$ time. We start with representation with $\mathcal{O}(n \log \log n)$ bits and at each step, we shorten the representation by at least one bit. This gives the running time of $\mathcal{O}(n^{1+c} \log \log n)$.

To get the running time of $\mathcal{O}(n^{1+\epsilon}\log\log n)$, we consider the graph $H_{\leq n^{\epsilon'}}$ for some sufficiently small epsilon. Note that the maximum out-degree of $H_{\leq n^{\epsilon'}}$ is $\mathcal{O}(n^{c\epsilon'})$. We can fix $\epsilon'$ small enough to decrease the maximum out-degree to $n^\epsilon$. Therefore, by using the same algorithm as above on graph $H_{\leq n^{\epsilon'}}$ for $\epsilon'$ sufficiently small, we can get a running time of $\mathcal{O}(n^{1+\epsilon}\log\log n)$ for any fixed $\epsilon > 0$. The representations of edges corresponding to vertices not in the graph $H_{\leq n^{\epsilon'}}$ are not shrunk.

Note that the presumptions of Theorem \ref{thm:expand} are fulfilled by the edges corresponding to vertices in $H_{\leq n^\epsilon}$ and the obtained representation of graph $G' = (V(G), V(H_\leq n^\epsilon))$, is therefore compact. The edges not in $H_{\leq n^\epsilon}$ are then added, increasing some offsets. The representation of an offset of length at least $n^{\epsilon''}$ for $\epsilon'' > 0$ is never increased asymptotically by inserting edges since it already has $\Theta(\log n)$ bits. There are at most $\mathcal{O}(n^{\epsilon''})$ edges of $G'$ shorter than $n^{\epsilon''}$ that span any single inserted edge. Lengthening of offsets shorter than $n^{\epsilon''}$, therefore, contributes at most $\mathcal{O}(n^{1-\epsilon'} n^{\epsilon''} \log \log n) \in o(n)$ for some $\epsilon''$ sufficiently small. The inserted edges themselves have representations of total length $\mathcal{O}(n^{1-\epsilon'}\log n) \in o(n)$. Additional $o(n)$ bits are used after the insertion of edges and the representation, therefore, remains compact.
\qed 
\end{proof}

\section{Separator hierarchy}\label{sec:sep_hierarchy}

In this section, we prove two generalizations of the separator hierarchy theorem. Our proof is based on the proof from \cite{bib:planar_sep}. Most importantly, we show that the recursive separator theorem also holds if we want the regions to have small size on average and not in the worst case. We also prove the theorem for weighted separator theorem with weights on edges. We show that the natural generalization of our two generalizations does not hold by presenting a counterexample.

Since the two theorems are very similar and their proofs only differ in one step, we present them as one theorem with two variants and show only one proof proving both variants. The difference lies in the reason why the Inequality \ref{eq:second_step} holds. The following lemma and observation prove the inequality under some assumptions and they will be used in the proof of the theorem.
\begin{equation} \label{eq:second_step}
\frac{c'\gamma_w W}{r_1^{1-c}} + \frac{c'(1-\gamma_w) W}{r_2^{1-c}} \leq \frac{c' W_n}{r^{1-c}}
\end{equation}

\begin{obs} \label{obs:second_step}
The Inequality \ref{eq:second_step} holds for $r_1 = r_2 = r$.
\end{obs}

\begin{lemma} \label{lem:second_step}
The Inequality \ref{eq:second_step} holds for $\gamma_w = \gamma_n$ and $r_1$, $r_2$ and $r$ satisfying the following.
\begin{equation} \label{eq:average}
r = \frac{1}{\frac{\gamma_n}{r_1} + \frac{1-\gamma_n}{r_2}} = \frac{r_1r_2}{\gamma_n r_2 + (1-\gamma_n)r_2}.
\end{equation}
\end{lemma}
\begin{proof}
Let $\gamma = \gamma_w = \gamma_n$.
We simplify the inequality
$$\frac{\gamma}{r_1^{1-c}} + \frac{1-\gamma}{r_2^{1-c}} \leq \frac{1}{r^{1-c}}$$
for $r_1, r_2$ and $r$ satisfying the equality (\ref{eq:average}). By substituting for $r$ and rearranging the inequality, we get
$$\gamma r_1^{1-c} + (1-\gamma) r_2^{1-c} \leq (\gamma r_1 + (1-\gamma) r_2)^{1-c}$$

We substitute $r_2 = \lambda r_1$. Note that this holds for $\lambda = 1$ and that we may assume $r_1 \leq r_2$ by symmetry. Since the inequality holds for $\lambda = 1$, it is sufficient to show the inequality for $\lambda \geq 1$ with both sides differentiated with respect to $\lambda$. By differentiating both sides and simplifying the inequality, we get
$$(x - (\lambda-1)\gamma)^{-c} \geq x^{-c}$$
which obviously holds, since $\lambda \geq 1$ and $\gamma > 0$.
\end{proof}

\medskip \noindent
Now we proceed to prove the two generalizations of the recursive separator theorem. Note that in the following, $r$ is the average or maximum region size, depending on whether the graph is weighted or not.
 \begin{theorem} \label{thm:sep_hierarchy}
Let $G$ be a (possibly weighted) graph satisfying the $n^c$ separator theorem with respect to its weights and let $P$ be its recursive balanced separator partition. Then if either
\begin{enumerate}[(i)]
\item \label{itm:first} the graph in not weighted and $r$ is the average size of a region in the partition $P$, or
\item \label{itm:second} the graph is weighted and $r$ is the maximum size of a region in the partition $P$.
\end{enumerate}
Then the total weight of edges not contained inside a region of $P$ is $\mathcal{O}(W/r^{1-c})$, where $W$ is the total weight (resp. number if unweighted) of all edges of $G$.
 \end{theorem}

In this proof, let $w(S)$ be the total weight of the edges in $S$ with $w(e)$ denoting the weight of the single edge $e$. 

\begin{proof}
We use induction on the number of vertices to prove the following claim.

\begin{claim} \label{clm:induction_claim}
Let us have a recursive separator partition $P$ of $n$-vertex graph $G$ of average region size $r$. Then $w(E(G) \setminus \bigcup_{p \in P} p) < \frac{c'W}{r^{1-c}} - \frac{c''W}{n^{1-c}}$ for some $c'$ and $c''$.
\end{claim}

Before the actual proof of this claim, let us define some notation. Let $c$, $\alpha$ and $\beta$ be the constants from the separator theorem (recall that separator theorem ensures existence of a partition of $V(G)$ into two sets of size at least $\alpha V(G)$ with edges of total weight at most $\beta \frac{W}{n^{1-c}}$ going across). Let $B(W, n, r)$ be the maximum value of $w(E(G) \setminus \bigcup_{p \in P} p)$ over all $n$-vertex graphs of total weight $W$ and all their recursive separator partitions with average region size $r$. We use $\gamma_n$ to denote a fraction of the number of vertices and $\gamma_w$ to denote a fraction of the total weight.

\begin{proof}[Proof of the claim]
We defer the proof of the base case until we fix the constant $c'$.

By the separator theorem, $B(W, n, r)$ satisfies the following recurrence.
$$B(W, n, r) = 0\text{ for }n \leq r$$
$$B(W, n, r) \leq \beta \frac{W}{n^{1-c}} + \max_{\substack{\alpha \leq \gamma_n \leq 1-\alpha\\ \gamma_w \in [0,1]}} B(\gamma_w W, \gamma_n n, r_1) + B((1-\gamma_w) W, (1-\gamma_n) n, r_2)$$
where $r_1,r_2$ are the respective average region sizes in the two subgraphs. It, therefore, holds that
$r = \frac{1}{\frac{\gamma_n}{r_1} + \frac{1-\gamma_n}{r_2}} = \frac{r_1r_2}{\gamma_n r_2 + (1-\gamma_n)r_2}$.

From the inductive hypothesis, we get the first inequality of the following. The second inequality follows from the Observation \ref{obs:second_step} for the case \ref{itm:first} and from the Lemma \ref{lem:second_step} for the case \ref{itm:second}.
\begin{equation} \label{eq:induction_step}
B(W, n ,r) \leq \beta \frac{W}{n^{1-c}} + \frac{c'\gamma_w W}{r_1^{1-c}} + \frac{c'(1-\gamma_w) W}{r_2^{1-c}} - c'' \frac{W}{n^{1-c}}(\gamma_n^c + (1-\gamma_n)^c) \leq 
\end{equation}
$$ \leq \beta \frac{W}{n^{1-c}} + \frac{c' W_n}{r^{1-c}} - c'' \frac{W}{n^{1-c}}(\gamma_n^c + (1-\gamma_n)^c)$$

It holds that $\gamma_n^c + (1-\gamma_n)^c \geq 1 + \epsilon_\alpha$, where $\epsilon_\alpha > 0$ is a constant depending only on $\alpha$, since $\gamma_n \in [\alpha, 1-\alpha]$ for $\alpha > 0$. We can therefore set $c''$ such that 
$$c'' \frac{W}{n^{1-c}}(\gamma_n^c + (1-\gamma_n)^c) - \beta \frac{W}{n^{1-c}} \geq c'' \frac{W}{n^{1-c}}$$
This completes the induction step.

\smallskip \noindent
For $c'$ large enough, the claimed bound in the base case is negative and it, therefore, holds.
\qed 
\end{proof}
\end{proof}

We conclude this section by showing that the following natural generalization of Theorem \ref{thm:sep_hierarchy} does not hold:
\begin{theorem}\label{thm:false}
The following generalization does \emph{not} hold: Let $G$ be a weighted graph satisfying the $n^c$ separator theorem with respect to its weights and let $P$ be its recursive separator partition. Let $r$ be the average size of a region in the partition $P$. Then the total weight of edges not contained in a region of $P$ is $\mathcal{O}(W/r^{1-c})$, where $W$ is the total weight of all edges of $G$.
\end{theorem}

\begin{proof}
We show that there is a weighted graph satisfying the $n^c$-separator theorem with respect to its weight and a recursive partition $P$ of $G$ with edges going between partition regions of $P$ that have total weight $\Theta(W)$, where $W$ is the total weight of all edges, and with average region size of $\Theta(n / \log n)$.
        
Let $G$ be an unweighted graph of bounded degree satisfying the $n^c$-separator theorem. We set weights of all its edges to be $1$, except for one arbitrary edge $e$ with weight $m - 1$, where $m$ is the number of edges of $G$. Note that $w(e) = W/2$. We denote this weighted graph by $G_w$.

Let $S$ be a separator in $G$ from the separator theorem. We modify $S$ in order to obtain a balanced separator $S_w$ in $G_w$ of weight $\mathcal{O}(W/n^{1-c})$. If $e \not \in S$, we set $S_w = S$. Otherwise, we remove $e$ from $S$ and add all other edges incident to its endpoints. This gives us $S_w$ which is a separator and its weight differs from the weight of $S$ only by an additive constant, since the graph $G$ has bounded degree. It follows that $G_w$ satisfies the $n^c$-separator theorem with respect to its weights.

We consider a partition $P$ constructed by the following process. Let $S$ be a separator from the separator theorem on $G_w$, partitioning $V(G_w)$ into vertex sets $A$ and $B$. If $e \in S$, we stop and set $A$ and $B$ as the regions of $P$. Otherwise, without loss of generality, $e \in A$. We set $B$ as a region of $P$ and recursively partition $A$.

At the end of this process, we get $P$ with edges of total weight at least $W/2$ between regions (as $e$ is not contained within any region). The partition $P$ has $\Theta(\log n)$ regions, so the average region size is $\Theta(n / \log n)$.
\qed
\end{proof}

\section{Representation for Paths in Trees}\label{sec:trees}

In this section, we show a linear algorithm that computes a cache-optimal layout of a given tree. We are assuming that the vertices have unit size and $B$ is the number of vertices that fit into a memory block. The same assumption has been used previously by \citet{bib:opt_average_lay}. This is a reasonable assumption for trees of fixed degree and for trees in which each vertex only has a pointer to its parent. It does not matter in which direction the paths are traversed and we may, therefore, assume that the paths are root-to-leaf.

We also show that it is \emph{NP}-hard to find an optimal compact layout of a tree and show an algorithm which gives a compact layout with I/O complexity at most $OPT + 1$.

\begin{definition}{Laid out tree:}
A laid out tree is an ordered triplet $T = (V, E, L)$, where $(V,E)$ is a rooted tree and $L: V \rightarrow \{0,1,2,\cdots,|V|\}$ assigns to each vertex the memory block that it is in. We require that at most $B$ vertices are assigned to any block. We treat the block 0 specially as the block already in the cache.
\end{definition}
We define $c_L'(P) = |\{L(v)\text{ for }v \in P\}\setminus\{0\}|$ to be the cost of path $P$ in a given layout $L$.
We define $c(T, k)$, the worst-case I/O complexity given $k$ free slots, as 
$$c(T, k)=\min_L(\max_P(c(P)))$$
where $P$ ranges over all root-to-leaf paths and $L$ over all layouts that assign at most $k$ vertices to block $0$. Since block $0$ is assumed to be already in cache, accessing these vertices does not count towards the I/O complexity.
We define $c(T)$, the worst-case I/O complexity of laid out tree $T$, to be $c(T, 0)$. This means $c(T)$ is the maximum number of blocks on a root-to-leaf path.
We define a worst-case optimal layout of a tree $T$ given $k$ free memory slots as a layout attaining $c(T,k)$.

We can observe that $c(T) \leq 1+\max_{u \in \delta(r_T)}(c(T_u))$. From the lemmas below follows that $c(T)$ only depends on the subtrees rooted in children of $r_T$ with the maximum value of $c(T_u)$.

\begin{lemma} \label{lem:tree_diff_by_one}
For any $k_1, k_2 \in [B]$, $|c(T, k_1) - c(T, k_2)| \leq 1$ and $c(T,k)$ is non-increasing in $k$.
\end{lemma}
\begin{proof}
The function $c(T,k)$ is monotonous in $k$ since a layout given $k_1$ free slots is a valid layout given $k_2$ slots for $k_2 \geq k_1$. Moreover $c(T, 0) = c(T, B) - 1$, since we can map vertices in the root's block to block $0$ instead. From this and the monotonicity, the lemma follows.
\qed 
\end{proof}

We define \textit{deficit} of a tree $k(T) = min\{k,\text{ such that }c(T, k) < c(T, 0)\}$. Note that $k(T) \leq B$. It follows from Lemma~\ref{lem:tree_diff_by_one} that $c(T,k') = c(T,0) = c(T,B)+1$ for all $k' < k(T)$ and $c(T,k') = c(T,0)-1 = c(T,B)$ for $k' \geq k(T)$.

\begin{lemma} \label{lem:root_in_zero}
For $k \geq 1$, there is a worst-case optimal layout attaining $c(T,k)$ such that root is in block $0$.
\end{lemma}
\begin{proof}
Let $L$ be a layout that does not assign block $0$ to the root. If no vertex is mapped to block $0$, we can move root to block $0$. Since block $0$ does not count towards I/O complexity, doing this can only improve the layout. Otherwise, let $v$ be vertex, which is mapped to block $0$. We construct layout $L'$ such that $L'(v) = L(r)$, $L'(r) = L(v)$ and $L'(u) = L(u)$ for all other vertices $u$. For any path $P$, $c'_L(P) \geq c'_{L'}(P)$, since any path which contains $v$ in layout $L'$ already contained it in $L$ and block $0$ does not count towards the I/O complexity.
\qed 
\end{proof}

It is natural to consider layouts in which blocks form connected subgraphs. This motivates the following definition

\begin{definition}
A partition of a rooted tree is \textit{convex} if the intersection of any root-to-leaf path with any set of the partition is a (possibly empty) path.
\end{definition}

Let $M_v$ be the set of successors $u$ of vertex $v$ with maximum value of $c(T_u)$.

\begin{lemma} \label{lem:trees_recurrence}
The function $c(T,k)$ satisfies the following recursive formula for $k\geq 1$.
$$c(T, k)=\min_{\{k_u\}} \max_{u\in M_v}c(T_u, k_u)$$ where the $\min$ is over all sequences $\{k_u\}$ such that $\sum_{u\in \delta(v)}k_u=k-1$.
\end{lemma}
\begin{proof}
By lemma~\ref{lem:root_in_zero}, we may assume that an optimal layout attaining $c(T,k)$ for $k \geq 1$ puts the root to block $0$ and allocates the remaining $k-1$ slots of block $0$ to root's subtrees, $k_u$ slots to the subtree $T_u$. On the other hand, from values of $k_u$, we can construct a layout with cost $\max_{u \in M_v}(c(T_u,k_u))$. 
\qed 
\end{proof}

\begin{problem}
\strut \\
Input: Rooted tree $T$\\
Output: Worst-case optimal memory layout of $T$.
\end{problem}

\begin{theorem} \label{thm:tree_layout}
There is an algorithm which computes a worst-case optimal layout in time $\mathcal{O}(n)$. Moreover, this algorithm always outputs a convex layout.
\end{theorem}
\begin{proof}
We solve the problem using a recursive algorithm. For each vertex, we compute $k(T_v)$ and $c(T_v)$. First, we define $d(T)$ and $c_{max}(v)$.
\begin{align*}
d(T_v) = 1 + \sum_{u \in M_v}{k(u)}, && c_{max}(v) = \max_{u \in \delta(v)}(c(T_u))
\end{align*}

If $d(T) < B$, we let $k(T_v) = d(T)$ and $c(T_v) = c_{max}(v)$. Otherwise $k(T_v) = 1$ and $c(T_v) = c_{max}(v)+1$. As a base case, we use that $c(T,k)=0$ when $|V(T)| \leq k$. For $k=0$, we use that $c(T, 0)=c(T, B) + 1$.

Using the values $k(T_u)$ and $c(T_u)$ calculated using the above recurrence, we reconstruct the worst-case optimal layout in a recursive manner. When laying out a subtree given $k$ free slots, we check whether $k \geq d(T)$. If it is, we distribute the $k-1$ empty slots (one is used for the root) in a way that subtrees $T_v$ for $v \in M(r_T)$ get at least $k(T_v)$ empty slots. Otherwise, distribute them arbitrarily. We put the root of a subtree into a newly created block if the subtree gets $0$ free slots. Otherwise, we put the root into the same block as its parent. It follows from the way we construct the solution that it is convex.

It follows from lemmas~\ref{lem:tree_diff_by_one} and~\ref{lem:trees_recurrence} that $c(T,k) = c(T,0)-1$ if and only if $k-1$ free slots can be allocated among the subtrees $T_u, u \in \delta(r_T)$ such that subtree $T_u$ gets at least $k(T_u)$ of them. It can be easily proven by induction that the algorithm finds for each vertex the smallest number of free slots required to make the allocation possible and calculates the correct value of $c(T_v)$.
\qed 
\end{proof}

If the subtree sizes are computed beforehand, we spend $deg(v)$ time in vertex $v$. By charging this time to the children, we show that the algorithm runs in linear time.

This algorithm can be easily modified to give a compact layout which ensures I/O complexity of walking on a root-to-leaf path to be at most $c(T) + 1$. This is especially relevant since finding the worst-case optimal layout is \textbf{NP}-hard, as we show in section~\ref{section:compact_hardness}. The algorithm can be modified to give a compact layout by changing the reconstruction phase such that we never give more than $|V(T_v)|$ free slots to the subtree of $T$ rooted in $v$ unless $k > |V(T)|$. Note that only the last block on a path can have unused slots. We can put blocks which are not full consecutively in memory, ignoring the block boundaries. Any path goes through at most $c(T)$ blocks out of which at most one is not aligned, which gives total I/O complexity of $c(T) + 1$.

The following has been proven before in \cite{bib:tree_asymp_lay} and follows directly from Theorem~\ref{thm:tree_layout}.

\begin{corollary} \label{lemma:convexity_opt}
For any tree $T$, there is a convex partition of $T$ which is worst-case optimal.
\end{corollary}
\begin{proof}
The corollary follows from Theorem~\ref{thm:tree_layout}, since the algorithm given in the proof is correct and always gives a convex solution.
\qed 
\end{proof}

Since the layout computed by the algorithm is always convex, we never re-enter a block after leaving it. This means that $c(T)$ really is the worst-case I/O complexity.

\smallskip
Finally, we show how to construct a compact representation with similar properties. Note that we do not claim $I/O$ optimality among all compact representations but only relative to the tree layout optimality as in Theorem~\ref{thm:tree_layout}.

\begin{theorem}\label{thm:tree-compact}
For a given tree $T$ with $q$ bits of extra data per vertex, there is a compact memory representation of $T$ using $\mathcal{O}(nq)$ bits of memory requiring at most $OPT_L$ I/O operations for leaf-to-root paths in general trees and root-to-leaf paths in bounded degree $d$ trees. Here $OPT_L$ is the I/O complexity of the optimal layout from Theorem~\ref{thm:tree_layout} when we set the vertex size to be $q+2\log n$ for leaf-to-root paths, or to $q+2d\log n$ for root-to-leaf paths.
\end{theorem}
\begin{proof}
The theorem is an indirect corollary of Theorems~\ref{thm:tree_layout} and~\ref{thm:expand}. We set the vertex size as indicated in the theorem statement (depending on the desired direction of paths) and obtain an assignment of vertices to blocks by Theorem~\ref{thm:tree_layout}. We call the set of the blocks $D$. Note that for $q=\Omega(\log n)$, this is already a compact representation.

For smaller $q$, we construct an auxiliary tree $T'$ on the blocks $D$ representing their adjacency in $T$. We can assume that $T'$ is a tree due to the convexity of the blocks of $D$. We apply the separator decomposition to obtain an ordering $R$ of $V_{T'}$ with short representation of offset edge representation (Lemma~\ref{lem:small_differences}). Similarly, we can get an ordering for each block in $D$. We order the vertices of $T'$ according to $R$, ordering the vertices within blocks according to orderings of the individual blocks. We obtain an ordering having offset edge representation of total length $\mathcal{O}(n \log q)$, as there is $\mathcal{O}(n/B)$ edges going between blocks with offset edge representations of total length $\mathcal{O}(n \log B \log q/B)$ and edges within blocks with offset edge representations of total length $\mathcal{O}(n \log q)$.

We now apply Theorem~\ref{thm:expand} on the edge offsets still split in memory blocks according to $D$, obtaining a bit-offset edge representation where the vertex representation of every block of $D$ still fits within one memory block, as we have previously reserved $2\log n+\Theta(1)$ memory for every pointer and $w_i\leq 1+2\log n$.
We merge consecutive blocks whose vertices fit together into one block. This ensures that every block has at least $B/2$ vertices.
\qed 
\end{proof}

\subsection{Hardness of worst-case optimal compact layouts} \label{section:compact_hardness}
In this section, we prove that it is \textbf{NP}-hard to find a worst-case optimal compact layout (that is, the packing with minimum I/O complexity out of all compact layouts). We show this by reduction from the 3-partition problem, which is strongly \textbf{NP}-hard \cite{bib:3-partition} (i.e.\ it is \textbf{NP}-hard even if all input numbers are written in unary).
\begin{problem}[3-partition]
\strut \\
\textbf{Input:} Natural numbers $x_1, \cdots, x_n$.\\
\textbf{Output:} Partition of $\{x_i\}_1^n$ into sets $Y_1, \cdots, Y_{n/3}$ such that $\sum_{x \in Y_i}{x} = 3(\sum_1^n{x_i})/n = S$ for each $i$.
\end{problem}

\begin{theorem}\label{thm:compact-hard}
It is \textbf{NP}-hard to find a worst-case optimal compact layout of a given tree $T$.
\end{theorem}
\begin{proof}
We let $B = S$. We construct the following tree. It consists of a path $P = p_1p_2\cdots p_B$ of length $B$ rooted in $p_1$. For each number $x_i$ from the 3-partition instance, we create a path of length $x_i$. We connect one of the end vertices of each of these paths to $p_B$.

Next, we prove the following claim. There is a layout of I/O complexity 2 iff the instance of 3-partition is a yes instance. We can get such layout from a valid partition easily by putting in a memory block exactly the paths corresponding to $x_i$'s that are in the same partition set. For the other implication, we first prove that $P$ is stored in one memory block. If it were not, we would visit at least two different memory block while traversing $P$ and there would be a root-to-leaf path that would visit three memory blocks. If $P$ is stored in one memory block, the I/O complexity of the tree is 2 iff the paths $p_i$ can be partitioned such that ever no part is stored in multiple memory blocks. There is such partition iff the instance of 3-partition is a yes instance.
\qed 
\end{proof}

\section{Further research}

Finally, we propose several open problems and future research directions.

\smallskip
Experimental comparison of traditional graph layouts with the layouts presented in our work and layouts proposed in prior work could both direct and motivate further research in this area.

\smallskip
While we optimize the separable graph layout for random walks it is conceivable that a minor modification would also match the worst-case performance of the previous results. 

\smallskip
The worst-case performance of the algorithm for finding the bit-offsets in Section~\ref{sec:bit-offsets} is most likely not optimal, and we suspect that the practical performance would be much better.

\smallskip
For the sake of simplicity, both our and prior representations of trees assume fixed vertex size (e.g. implicitly in the results on layouts) or allow $q=\mathcal{O}(1)$ extra bits per vertex in the compact separable graph representation. This could be generalized for vertices of different sizes and unbounded degrees.
    


\bibliographystyle{splncsnat}
\bibliography{bibliography}

\end{document}